\RequirePackage{fix-cm}
\documentclass[twocolumn]{svjour3}          % twocolumn
 \relax
\smartqed  % flush right qed marks, e.g. at end of proof
\usepackage{graphicx}
\usepackage{verbatim}
\usepackage[authoryear]{natbib}
\usepackage{amsmath}
\journalname{Statistics and Computing}
\DeclareMathOperator{\xgamma}{\Gamma}
\DeclareMathOperator{\xf}{f}
\newcommand{\design}{{\mathcal D}}
\newcommand{\fraction}{{\mathcal F}}

\begin{document}

\title{Random generation of optimal saturated designs %\thanks{Grants or other notes
%about the article that should go on the front page should be
%placed here. General acknowledgments should be placed at the end of the article.}
}
\subtitle{An approach based on discovery probability}

%\titlerunning{Short form of title}        % if too long for running head

\author{Roberto Fontana}

%\authorrunning{Short form of author list} % if too long for running head

\institute{R. Fontana \at
              Department of Mathematical Sciences - Politecnico di Torino \\
              Tel.: +39-011-0907504
              \email{roberto.fontana@polito.it}           %  \\
%             \emph{Present address:} of F. Author  %  if needed
}

\date{Received: date / Accepted: date}
% The correct dates will be entered by the editor

\maketitle

\begin{abstract}
Efficient algorithms for searching for optimal saturated designs are widely available. They maximize a given efficiency measure (such as D-optimality) and provide an optimum design. Nevertheless, they do not guarantee a \emph{global} optimal design. Indeed, they start from an initial random design and find a local optimal design. If the initial design is changed the optimum found will, in general, be different. A natural question arises. Should we stop at the design found or should we run the algorithm again in search of a better design? This paper uses very recent methods and software for discovery probability to support the decision to continue or stop the sampling. A software tool written in SAS has been developed.
% \PACS{PACS code1 \and PACS code2 \and more}
% \subclass{MSC code1 \and MSC code2 \and more}
\end{abstract}

\section{Introduction}
\label{intro}
In the design of experiments, optimal designs, or optimum designs, are a class of experimental designs that are optimal with respect to a given statistical criterion.

In this paper we focus on saturated optimum designs (SOD). Saturated designs contain a number of points that is equal to the number of parameters of the model. It follows that SODs are often used in place of standard designs, such as orthogonal fractional factorial designs, when the cost of each experimental run is high. Main references to this topic include \cite{atkinson2007optimum}, \cite{pukelsheim2006optimal}, \cite{shah1989theory} and \cite{wynn1970sequential}.
 
The optimality of a design depends on the statistical model that is assumed and is assessed with respect to a statistical criterion, which, for information-based criteria, is related to the variance-matrix of the model parameter estimators. Well-known and commonly used criteria are A-optimality and D-optimality.

Widely used statistical systems like SAS and R have procedures for finding an optimal design according to the user's specifications. In this paper we will refer to Proc Optex of SAS/QC (\cite{sasqc:10}), but the approach can be adopted for other software.

The Optex procedure searches for optimal experimental designs. The user specifies an efficiency criterion, a set of candidate design points, a linear model and the size of the design to be found and the procedure generates a subset of the candidate set so that the terms in the model can be estimated as efficiently as possible. By default, the standard output of the procedure is a list of $10$ designs that are found as the result of $10$ runs of the exchange search algorithm (\cite{mitchell1970use}) starting each time from an initial completely randomly chosen design. 

The number of times that we decide to run the search algorithm is crucial. Obviously, if we increase it, in general we will explore different local optima with the possility to find better designs. On the other hand, sometimes, the extra time that we use to explore other possibilities is wasted because new optima do not exist. This work aims at developing a methodology that could support the user in making the decision whether to stop or continue the search.

\begin{comment}
Most experimental situations call for standard designs, such as fractional factorials, orthogonal arrays, central composite designs, or Box-Behnken designs. Standard designs have assured degrees of precision and orthogonality that are important for the exploratory nature of experimentation. In some situations, however, standard designs are not available, such as when 

not all combinations of the factor levels are feasible 

the region of experimentation is irregularly shaped 

resource limitations restrict the number of experiments that can be performed 

there is a nonstandard linear or a nonlinear model 
The Optex procedure can generate an efficient experimental design for any of these situations. 
\end{comment}

The paper is organized as follows. In Sect.~\ref{sec:1} we state the problem of finding new optimal designs as the problem of finding new species in a population. Then, in Sect.~\ref{sec:2}, using some examples, we describe how our methodology, which is based on the estimator of the discovery probability, could be used for optimal design generation. In Sect.~\ref{sec:algo} we describe the algorithm in more detail. The software code that has been developed is written in SAS, is available on request and can be used for any choice of factors, levels and model. Concluding remarks are in Sect.~\ref{sec:conclusion}.

\section{Optimal designs vs richness of species} \label{sec:1}
We consider the following setting that is quite common in optimal design problems.

We have $d$ factors, $A_1,\ldots, A_d$.  The factor $A_i$ has $s_i$ levels coded with the integer $0,\ldots,s_i-1$, $i=1,\ldots,d$. The full factorial design is $\mathcal{D}=\left\{0,\ldots,s_1-1\right\} \times \ldots \times \left\{0,\ldots,s_m-1\right\}$. For each point $\zeta=(\zeta_1,\ldots,\zeta_d)$ of $\mathcal{D}$ we consider a real-valued random variable $Y_{\zeta_1,\ldots,\zeta_d}$. We make the hypothesis that the means of the responses, $E\left[Y\right]$ where $Y$ is the column vector $\left[ Y_\zeta; \zeta \in \design \right]$
%\[
%\left[ Y_{\zeta_1, \ldots, \zeta_d}; \; \zeta_i \in \left[s_i\right], i=1,\ldots,d \right] \, ,
%\]
can be modeled as
\begin{equation} \label{eq:1}
E\left[Y\right]=X_{\mathcal{D}} \beta \, ,
\end{equation} 
where $X_{\mathcal{D}}$ is the non-overparametrized design matrix, as it will be defined in Sect.~\ref{subsec:opt}, and $\beta$ is the subset of all the effects (constant effect, main effects and interactions) that are supposed to affect the response $Y$.  

Given an efficiency criterion $\phi$, a saturated optimal design ($\phi$-SOD) is a subset of the full factorial design $\mathcal{D}=\left\{0,\ldots,s_1-1\right\} \times \ldots \times \left\{0,\ldots,s_m-1\right\}$, whose size is equal to the number of degrees of freedom of the model (\ref{eq:1}) and that maximizes this criterion $\phi$. In this paper we focus on information-based criteria and, in particular, on $D$-optimality but other criteria can be chosen (like $A$-optimality and $G$-optimality). We denote this type of problem with the triple $(\design, \mathcal{M},\phi)$ where $\design$ is the full design, $\mathcal{M}$ is the hypothesized model (see Eq.~\ref{eq:1}) and $\phi$ is the optimality criterion. 

Given a subset $\fraction$ of $\design$, the information matrix is defined as $X_\fraction'X_\fraction$ where $X_\fraction$ is the design matrix corresponding to $\fraction$ and $X'$ is the transpose of $X$. $D$-optimality aims at maximizing $D_\fraction$, the determinant of the information matrix
\begin{equation} \label{dopt}
D_\fraction=\det(X_\fraction'X_\fraction) \, .
\end{equation} 

There are several algorithms for searching for $D$-optimal designs. They have a common structure. They start from an initial design, randomly generated or user specified, and move, in a finite number of steps, to a better design. In general, if a different initial design is chosen, a different optimal design is found. 

It follows that, given an algorithm $\alpha$, a population $\mathcal{A}_\alpha^D$ of $D$-optimal designs can be defined. This population is made up of all the saturated designs that are the result of the execution of the algorithm $\alpha$ and is a subset of all the subsets of $\design$ of size equal to the number of degrees of freedom of the model.

The elements of $\mathcal{A}_\alpha^D$ can be classified into species, according to the criterion for which $\fraction_1 \in \mathcal{A}_\alpha^D$ and $\fraction_2 \in \mathcal{A}_\alpha^D$ are of the same species if and only if they have the same value in terms of the $D$ criterion, $D_{\fraction_1}=D_{\fraction_2}$.
  
We observe that, as proved in Proposition~\ref{pr:iso}, isomorphic designs belong to the same species, while, in general, the opposite is not true because there are designs with the same value of the $D$ criterion but that are not isomorphic. As is known two designs are isomorphic if one can be obtained from the other by relabeling the factors, reordering the runs, and switching the levels of factors, e.g. \cite{clark2001equivalence}. 

\begin{proposition} \label{pr:iso}
Let us consider $\fraction_1 \subseteq \design$ and $\fraction_2 \subseteq \design$. If $\fraction_1$ and $\fraction_2$ are isomorphic then $D_{\fraction_1}=D_{\fraction_2}$.
\end{proposition}
\begin{proof}
We separately analyse row/column permutations and the switching of the levels of some factors.
If $\fraction_2$ is obtained permuting the rows and/or the columns of $\fraction_1$ it follows that
\[
X_{\fraction_2}=R X_{\fraction_1} C
\] 
where $R$ and $C$ are permutation matrices. Then
\begin{eqnarray*}
D_{\fraction_2}=\\ 
= \det((X_{\fraction_2}'X_{\fraction_2}))= (\det(R))^2 \det((X_{\fraction_1}'X_{\fraction_1})) (\det(C))^2 = \\ 
= D_{\fraction_1}
\end{eqnarray*}
being $\det(R)=\det(C)=1$. 
A similar argument holds for switching the levels of some factors. \qed  
\end{proof} 

\begin{comment}
If $\fraction_2$ is obtained by $\fraction_1$ switching the levels of factors we suppose, without loss of generality, that the factor $A_m$ has been conithe level $1$ of factor $A_m$ has been exchanged with the level $s_m$ of the same facori 
\end{comment}

Studying the species of $\mathcal{A}_\alpha^D$ or, in general, of $\mathcal{A}_\alpha^\phi$ where $\phi$ is an optimal criterion, is interesting for optimal design generation. Let us consider the problem $(\design, \mathcal{M},\phi)$ and let us choose an algorith $\alpha$ to search for $\phi$-SODs. If we run this algorithm $n$ times, each time starting from a completely random initial design, we will get a sample of $n$ elements of $\mathcal{A}_\alpha^\phi$. Such elements can be classified in $k_n \leq n$ different species according to the value of the criterion $\phi$. Recent methods for discovery probability estimation, \cite{favaro:12}, can be applied to the vector $(\ell_1,l_2,\ldots,\ell_n)$ where $\ell_r$ is the number of species in the sample with frequency $r$, $r=1,\ldots,n$. In particular, based on a sample of size $n$, for any additional unobserved sample size $m\geq 0$ and for any frequency $k=0,\ldots, n+m$, these methods provide, an explicit estimator for the probability $U_{n+m}(k)$ that the $(n+m+1)$-th observation coincides with a species whose frequency, within the sample size $n+m$, is exactly $k$. The case $m=k=0$ corresponds to assessing the probability of finding a new species in the subsequent observation, that in the context of optimal designs, is the probability of finding a saturated design with a different value of the criterion $\phi$ in the subsequent run of the algorithm. If this probability $U_{n+0}(0)$ is sufficiently high (let us say greater than $0.1$ or even $0.05$) it would be convenient to run the algorithm again because it is likely that we could find a new optimal design. If we found a new design, it could have a greater value of $\phi$ and this obviously represents an improvement to our optimization process. Even if this new design did not have an higher value of $\phi$ than the existing ones, this would give the possibility to increase the known part of $\mathcal{A}_\alpha^\phi$. In particular, for $D$-optimal designs, from Proposition~\ref{pr:iso}, we know that designs with different values of $D_\fraction$ are non-isomorphic designs. It is quite common, in practical applications, to choose a design where the optimal criterion has a slightly smaller value than the maximum obtained but which has other better characteristics, such as space filling properties. The knowledge of a set of non-isomorphic designs can also be used for non parametric testing procedures, \cite{arb_fon_rag} and \cite{basso04}. 

\subsection{The design matrix} \label{subsec:opt}
The design matrix $X_{\mathcal{D}}$ in Eq.~\ref{eq:1} is built as follows.
\begin{itemize}
\item The first column is equal to $1$ and corresponds to the constant effect, denoted by $\mu$. The constant effect is always considered as a term of the model.
\item If the main effect of the factor $A_i$ is to be considered in the model, the corresponding $s_i-1$ columns are computed as follows. For a design point with $A_i$ at its $k$-th level
\begin{itemize}
\item if $1 \leq k \leq s_i-1$ the columns are all $0$ except for the $k$-th column that is $1$;
\item if $k=s_i$ the columns are all $-1$  
\end{itemize}
\item If an interaction $A_{i_1} \star \ldots \star A_{i_k}$ is to be considered in the model, the corresponding $(s_{i_1}-1)\cdot \ldots \cdot (s_{i_k}-1)$ columns are computed by taking the horizontal direct product of the colummns corresponding to the main effects of $A_{i_1}, \ldots, A_{i_k}$.
\end{itemize}
This coding corresponds to modeling without over parametrization and $X_\design$ is full rank.

For a subset $\mathcal{F}$ of $\mathcal{D}$, the design matrix $X_\mathcal{F}$ is simply built deleting from $X_\mathcal{D}$ the rows that correspond to the points of $\mathcal{D}$ that are not in $\mathcal{F}$.

\subsection{Discovery probability} \label{subsec:dp}
We briefly summarize  the main results that are used in this work, as in \cite{favaro:12}. The interested reader should refer to the original paper for a detailed description of the methodology.

Given a sample of size $n$, $(\ell_1,\ldots,\ell_n)$, where $\ell_r$ is the frequency of species that have been observed $r$-times in the sample, $r=1,\ldots,n$. We have $\sum_{i=1}^n i \ell_i=n$. We denote the number of different species that have been observed in the sample by $j$. We get $\sum_{i=1}^n \ell_i=j$. 

Based on a sample of size $n$, for an additional unobserved sample size $m \geq 0$ and for any frequency $k=0, \ldots,n+m$, using a non parametric Bayesian approach, Favaro et al provide an estimator for the probability $U_{n+m}^k$ that the $(n+m+1)$-th observation coincides with a species whose frequency, within the sample of size $n+m$, is exactly $k$.

We are interested in discovering new species, that correspond to the case $k=0$.

From Section~2 of on p.1190 we obtain 
\[
U_{n+0}(0)=\frac{V_{n+1,j+1}}{V_{n,j}}
\]
where, for the two-parameter Poisson-Dirichlet process, we have $V_{n,j}=\prod_{i=1}^{j-1} (\theta+i\sigma)/(\theta+1)_{n-1}$, $\sigma \in (0,1)$, $\theta >- \sigma$. The symbol $(a)_n$ denotes the $n$-th ascending factorial of $a$,  $(a)_n=a(a+1)\ldots(a+n-1)$, $(a)_0 \equiv 1$.
It follows that 
\[
U_{n+0}(0)=\frac{\theta+j\sigma}{\theta+n}
\]
and, for $m>0$, we obtain
\[
U_{n+m}(0)=\frac{\theta+j\sigma}{\theta+n} \frac{ (\theta+n+\sigma)_m }{ (\theta+n+1)_m  } \, .
\]

The estimates $\hat{\sigma},\hat{\theta}$ of $\sigma,\theta$ are obtained as    
\begin{equation} \label{eq:tetasigma}
\arg \max_{(\sigma,\theta)} \frac{\prod_{i=1}^{j-1} (\theta+i\sigma)}{(\theta+1)_{n-1}} n! \prod_{i=1}^n \{ \frac {(1-\sigma)_{i-1}}{i!} \}^{\ell_i} \frac{1}{\ell_i!} \, .
\end{equation}

Using $(\hat{\theta},\hat{\sigma})$ we finally obtain the estimates of the discovery probability at the $(n+1)$-th observation
\begin{equation} \label{eq:dp}
\hat{U}_{n+0}(0)=\frac{\hat{\theta}+j\hat{\sigma}}{\hat{\theta}+n}
\end{equation}
and at the $(n+m+1)$-th observation, $m>0$,
\begin{equation} \label{eq:dpm}
\hat{U}_{n+m}(0)=\frac{\hat{\theta}+j\hat{\sigma}}{\hat{\theta}+n} \frac{ (\hat{\theta}+n+\hat{\sigma})_m }{ (\hat{\theta}+n+1)_m  }
\end{equation}

\begin{comment}
BNP_Good_Toulmin_discovery_estimate(i)=((theta+(j*sigma))/(theta+n))*exp(rising_factorial(theta+n+sigma,i)-rising_factorial(theta+n+1,i));
\end{comment}

\section{Methodology and Applications} \label{sec:2}
We  repeat the search for optimal designs to analyse the population $\mathcal{A}_\alpha^D$ of $D$-optimal designs that can be found for a given problem using a predefined algorithm. Each time the algorithm starts from a randomly chosen initial design. We set a maximum number of iterations equal to $M_{\star}$ and we continue the process until the estimate of the discovery probability at the subsequent observation goes under a given threshold $p_{\star}$  or the maximum number of iterations is reached.

The procedure can be described as follows. A problem $(\design,\mathcal{M},\phi)$, with $\phi=D$ in our examples, is defined and an algorithm $\alpha$ for $\phi$-optimal design generation is chosen. For each iteration $s$, $s=1,\ldots,M_{\star}$,
\begin{enumerate}
\item using the algorithm $\alpha$, a $\phi$-optimal saturated design $\fraction_s$ is obtained;
\item the values of the $\phi$-criterion of $\fraction_s$ is computed;
\item the vector $(\ell_1,\ldots,\ell_{s})$ is built, where $\ell_r$ is the number of species with frequency $r$, $r=1,\ldots,s$;
\item an estimate $(\hat{\sigma}_s,\hat{\theta}_s)$ is obtained, see Eq.~\ref{eq:tetasigma};
\item an estimate of $\hat{U}_{s+0}(0)$ is computed using Eq.~\ref{eq:dp};
\item if $\hat{U}_{s+0}(0)<p_{\star}$ the algorithm stops, otherwise the next iteration $s+1$ is performed (if $s+1>M_{\star}$ the algorithm stops).
\end{enumerate}
The main output of the algoritm is a set of designs, where each design belongs to a different species, i.e. has a different value of the $\phi$-criterion.

We show how the methodology works using the following problem. Let us consider $7$ factors, each with $2$ levels and the model that contains the overall mean, the main effects and all the 2-factor interactions for a total of $1+7+21=29$ degrees of freedom. We search for \emph{saturated} $D$-optimal designs that is $D$-optimal designs that contains $29$ points.  

We use Proc Optex \cite{sasqc:10} with the exchange method, which is its default search method. With the default setting, the algorithm starts from $10$ initial randomly chosen designs providing $10$ $D$-optimal designs. We consider the design with the highest value of the $D$-efficiency of the $10$ optimal designs as the optimal design found by the algorithm.

Setting the seed that is used for the random generation of the initial designs at $6789$, the best of the $10$ optimal designs, that we denote by $\fraction_1$, has $D_{\fraction_1}=9.0911E39$ and $E_{\fraction_1}^D=82.3162$, where $E_{\fraction}^D$, the $D$-efficiency of a $\fraction$, is defined as 
\[
E_{\fraction}^D = 100 \times \left( \frac{1}{{\#\fraction}} D_{\fraction}^{\frac{1}{{\#\fraction}}} \right)
\]
where $\#\fraction$ is the number of runs of $\fraction$ that coincides with the degrees of freedom of the model for saturated designs.

Now we run the procedure above with $M_{\star}=1,000$ and $p_{\star}=0.10$. 

After $493$ runs, the estimate of the discovery probability at the next observation becomes lower than $p_\star=0.10$ and the algorithm stops ($\tilde{U}_{493+0}(0) \approx 0.099)$. We find $103$ different local $D$-optimal designs. All these designs are not isomorphic (Proposition~\ref{pr:iso}). The maximum (minimum) value of $D$-efficiency is $85.6265$ ($78.9605$).

\begin{table}
% table caption is above the table
\caption{Number $\ell_r$ of $D$ optimal designs that have found $r$ times, $r=1,\ldots,493$; only $\ell_r \neq 0$ are shown.}
\label{tab:1}       % Give a unique label
% For LaTeX tables use
\begin{tabular}{rr}
\hline\noalign{\smallskip}
$r$ & $\ell_r$  \\
\noalign{\smallskip}\hline\noalign{\smallskip}
1 & 47 \\
2 & 18 \\
3 & 7 \\
4 & 10 \\
5 & 2 \\
6 & 4 \\ 
9 & 2 \\
11 & 1 \\
12 & 1 \\
14 & 2 \\
15 & 1 \\
16 & 1 \\
17 & 2 \\
20 & 1 \\
36 & 1 \\
39 & 1 \\
40 & 1 \\
46 & 1 \\
\noalign{\smallskip}\hline
Total & 103 \\
\noalign{\smallskip}\hline
\end{tabular}
\end{table}
  
%\begin{figure}[h]
% Use the relevant command to insert your figure file.
% For example, with the graphicx package use
%\includegraphics[scale=0.3]{2_7.png}
%\includegraphics[scale=0.8]{2_7.ps}
% figure caption is below the figure
%\caption{$(\ell_1,\ldots,\ell_{493})$}
%\label{fig:1}       % Give a unique label
%\end{figure}

We decide to continue the search for new species choosing $p_{\star}=0.05$ and $M_{\star}=2,000$. The latter value is chosen taking into account that using Eq.~\ref{eq:dpm} we get $\tilde{U}_{493+1000}(0)=0.049$ and $\tilde{U}_{493+2000}(0)=0.035$. We observe that these supplementary runs are added to the previous ones.

After $1,271$ supplementary runs the estimate of the discovery probability at the next observation becomes lower than $0.05$, $\tilde{U}_{1764+0}(0) \approx 0.0499$. After $1,271+493=1,764$ simulations we observe $191$ different $D$-optimal designs. The maximum value of $D$-efficiency is still $85.6265$, while the minimum is $78.1134$.

We can now use the Fedorov algorithm, \cite{fedorov1972theory}, that is considered more reliable, even if slower, than the exchange algorithm. We keep the standard setting for which, at each iteration, $10$ optimal designs are generated and the one among them that has the highest $D$-efficiency value is taken as the optimal design.

We choose $3456$ as the initial seed. The first iteration provides an optimal design $\fraction_1$ with $E_{\fraction_1}^D=82.7079$.  Now we repeat the procedure with $M_{\star}=1,000$ and $p_{\star}=0.10$. After only $18$ iterations, as $\tilde{U}_{18+0}(0) \approx 0.087$, the algorithm stops, with $4$ different designs. The maximum (minimum) value of $D$-efficiency is $83.9844$ ($82.4212$). We have empirical evidence that the Fedorov algorithm is more stable than the exchange algorithm. We observe that the best design found with the exchange algorithm, that has $D$-efficiency equal to $85.6265$, is not found in this first sample. We were able to find it running the algorithm again with $M_\star=1,000$ and $p_\star=0.01$.

\section{The algorithm} \label{sec:algo}
In this section we provide a detailed description of the algorithm that has been developed to study the population $\mathcal{A}_\alpha^D$ that contains all the $D$-optimal designs that can be found by the algorithm $\alpha$. 

A problem $(\design,\mathcal{M},\phi=D)$ is defined and an algorithm $\alpha$ for $D$-optimal design generation is chosen. 
The set of candidates that, in our setting, is the full factorial design is generated using an ad-hoc module written in SAS/IML.
The algorithm $\alpha$ can be chosen from a list of methods that includes the exchange algorithm and the Fedorov algorithm.

For each iteration $s$, $s=1,\ldots,M_{\star}$,
\begin{enumerate}
\item using the algorithm $\alpha$, a $D$-optimal saturated design $\fraction_s$ is obtained;
\item the value of the $D$-efficiency, $E_{\fraction_s}^D$, of $\fraction_s$ is computed;
\item the vector $(\ell_1,\ldots,\ell_s)$ is built, where $\ell_r$ is the number of species with frequency $r$, $r=1,\ldots,s$;
\item an estimate $(\hat{\sigma}_s,\hat{\theta}_s)$ is obtained, see Eq.~\ref{eq:tetasigma};
\item an estimate of $\hat{U}_{s+0}(0)$ is computed using Eq.~\ref{eq:dp};
\item if $\hat{U}_{s+0}(0)<p_{\star}$ the algorithm stops, otherwise the next iteration $s+1$ is performed (if $s+1>M_{\star}$ the algorithm stops).
\end{enumerate}

The main output of the algoritm is a set of designs, where each design belongs to a different species, i.e. has a different value of the $D$-criterion.

\subsection{Steps 1 and 2}
At iteration $s$, with the chosen algorithm $\alpha$, the Proc Optex procedure is used to generate a $D$-optimal design, $\fraction_s$. The species of $\fraction_s$ is the value of its $D$-efficiency, $E_{\fraction_s}^D$. The value of the efficiency is rounded to four decimal digits to avoid creating different species from numerical effects.  

\subsection{Step 3}
Using all the designs $\fraction_1, \ldots, \fraction_s$ with their corresponding $D$-efficiencies, $E_{\fraction_1}^D, \ldots, E_{\fraction_s}^D$ the vector $(\ell_1,\ldots,\ell_s)$ is built, where $\ell_r$ is the number of species with frequency $r$, $r=1,\ldots,s$.

\subsection{Step 4}
An estimate $(\hat{\sigma}_s,\hat{\theta}_s)$ must be obtained searching for $(\sigma,\theta)$, $\sigma \in (0,1)$, $\theta>-\sigma$ that maximizes $\xf(\sigma,\theta)$, (see Eq.~\ref{eq:tetasigma}),
\[
\xf(\sigma,\theta)=\frac{\prod_{i=1}^{j-1} (\theta+i\sigma)}{(\theta+1)_{n-1}} n! \prod_{i=1}^n \{ \frac {(1-\sigma)_{i-1}}{i!} \}^{\ell_i} \frac{1}{\ell_i!}
\]

The Genetic Algorithm module of SAS/IML has been used. In order to manage the constraints $\sigma \in (0,1)$, $\theta>-\sigma$ the search has been performed in the region $\mathcal{R=}\left[\delta, 1-\delta\right] \times \left[-(1-\delta),T_M\right]$ with $\delta=0.01$ and $T_M=1,000$. This region contains the non-feasible region made by the points inside the simplex $\mathcal{S}=\mathcal{R}\cap \{(\sigma,\theta): \theta \leq -\sigma\}$ whose vertices are $(\delta,-(1-\delta))$, $(\delta,-\delta)$ and $(1-\delta,-(1-\delta))$.
We observe that the edges of $\mathcal{S}$ contain non-feasible points.

We decided to manage this constraint with the penalty method, because this method usually works well when most of the points in the solution space do not violate the constraints, as in our problem. The way in which the penalty in the objective function for unsatisfied constraints has been imposed is described here. 

From the point of view of the search of the point $(\sigma_\star,\theta_\star)$ that maximizes $f(\sigma,\theta)$, it is equivalent to consider $\log \xf(\sigma,\theta)$ instead of $\xf(\sigma,\theta)$
\begin{eqnarray*} 
\log \xf(\sigma,\theta)=\log(\prod_{i=1}^{j-1} (\theta+i\sigma))+\log(n!)+ \\
-\log((\theta+1)_{n-1})+\log(\prod_{i=1}^n \{ \frac {(1-\sigma)_{i-1}}{i!} \}^{\ell_i}) -\log(\ell_i!) \, .
\end{eqnarray*}
Omitting the terms that do not depend on $\sigma$ and $\theta$ and as $(a)_n=\frac{\xgamma(a+n)}{\xgamma(a)}$ where $\xgamma$ is the gamma function, the previous equation becomes the function $\xf_\star(\sigma,\theta)$ here
\[
\xf_\star(\sigma,\theta)=\xf_\star^{(1)}(\sigma,\theta)+\xf_\star^{(2)}(\sigma,\theta) \, ,
\]
where
\[
\xf_\star^{(1)}(\sigma,\theta)=\sum_{i=1}^{j-1}\xf_\star^{(1,i)}(\sigma,\theta)
\]
with $\xf_\star^{(1,i)}(\sigma,\theta)=\log(\theta+i\sigma)$ and
\begin{eqnarray*}
\xf_\star^{(2)}(\sigma,\theta)=-\log\xgamma(\theta+n)+\log\xgamma(\theta+1)+\\
+\sum_{i=1}^n \ell_i\log\xgamma(i-\sigma)-j\log\xgamma(1-\sigma) \, .
\end{eqnarray*}
We observe that, if the point $(\sigma,\theta)\in \mathcal{R}$ does not satisfy the constraint $\theta > -\sigma$ only  $\xf_{\star}^{(1)}(\sigma,\theta)$ becomes not defined. We apply a penalty value to $\xf_{\star}^{(1)}(\sigma,\theta)$ and to $\xf_{\star}^{(2)}(\sigma,\theta)$ as described below.

Given a point $P_1$ in the non-feasible region, $P_1 = (\sigma,\theta) \in \mathcal{S}$, $\tilde{P}_1$, the closest point to $P_1$ with respect to the euclidean distance that lies in the feasible region, is determined
\[
\tilde{P}_1 = (\tilde{\sigma},\tilde{\theta})=(\frac{1}{2} (\sigma - \theta + \epsilon), \frac{1}{2} (\theta - \sigma + \epsilon))
\]
where $\epsilon$ is a very small number to ensure that $\tilde{P}_1$ is feasible, i.e. $\tilde{P}_1 \in \mathcal{R} \cap \overline{\mathcal{S}}$. We used $\epsilon=0.001$. The value of the function $\xf_\star^{(1,1)}$ is computed in $\tilde{P}_1$ getting $\tilde{Y}_1=\xf_\star^{(1,1)}(\tilde{\sigma},\tilde{\theta})=\log{\epsilon}$. Then the value $Y_1$ of $\xf_{\star}^{(1,1)}$ in $P_1$ is defined as $\xf_\star^{(1,1)}(\sigma,\theta)=(1+b_1)\tilde{Y}_1$ where $b_1$ is the euclidean distance between $P_1$ and $\tilde{P}_1$, $b_1=\sqrt{\frac{1}{2}(\sigma+\theta-\epsilon)^2}$. In an analogous way, we apply this penalty method to all $P_i=(i\sigma,\theta)$ that eventually fall in the non-feasible region $\mathcal{S}$ getting $\xf_{\star,P}^{(1)}(\sigma,\theta)$, the penalized version of $\xf_{\star}^{(1)}(\sigma,\theta)$,
\[
\xf_{\star,P}^{(1)}(\sigma,\theta)=\sum_{i=1}^{j-1} \xf_{\star}^{(1,i)}(\sigma,\theta)
\]
where
\[
\xf_{\star}^{(1,i)} =
\begin{cases}
\log(\theta+i\sigma) & \text{if } \theta+i\sigma>0 \\
(1+b_i)\log(\epsilon) & \text{if } \theta+i\sigma \leq 0  
\end{cases} , \,
i=1,\ldots,j-1
\, ,
\]  
and $b_i$ is the euclidean distance between $P_i=(i\sigma,\theta)$ and $\tilde{P}_i=(\frac{1}{2}(i\sigma-\theta+\epsilon,\frac{1}{2}(\theta-i\sigma+\epsilon)$ determined as described above. 
The penalized version $\xf_{\star,P}^{(2)}(\sigma,\theta)$ of $\xf_{\star}^{(2)}(\sigma,\theta)$ is simply defined as
\[
\xf_{\star,P}^{(2)}(\sigma,\theta)=
\begin{cases}
\xf_{\star}^{(2)}(\sigma,\theta) & \text{if } \theta+\sigma>0 \\
(1+b_1)\xf_{\star}^{(2)}(\sigma,\theta) & \text{if } \theta+i\sigma \leq 0 \\
& \text{ and } \xf_{\star}^{(2)}(\sigma,\theta) \leq 0 \\
(1-b_1)\xf_{\star}^{(2)}(\sigma,\theta) & \text{if } \theta+i\sigma \leq 0 \\
& \text{ and } \xf_{\star}^{(2)}(\sigma,\theta) > 0 
\end{cases}
\, .
\]   
We observe that
\begin{enumerate}
\item $p<q \Rightarrow b_p > b_q \; p,q =1,\ldots,j-1$;
\item $b_1 \leq \frac{\sqrt{2}}{2}(1+\epsilon -2\delta)$. For $\delta=0.01$ and $\epsilon=.001$ we get $b_1<0.694$. 
\end{enumerate}
Using the penalty method, an estimate $(\hat{\sigma}_s,\hat{\theta}_s)$ is obtained finding the maximum of $\xf_{\star,P}(\sigma,\theta)=\xf_{\star,P}^{(1)}(\sigma,\theta)+\xf_{\star,P}^{(2)}(\sigma,\theta)$.

\subsection{Steps 5 and 6}
The estimate of the discovery probability at the next iteration, $\hat{U}_{s+0}(0)$, is computed as described in Sect.~\ref{sec:2}, Eq~\ref{eq:dp}. If its value is lower than $p_{\star}$ the algorithm stops, otherwise the next iteration $s+1$ is performed (if $s+1 > M_\star$ the algorithm stops).  

\section{Conclusion} \label{sec:conclusion}
Given an optimality crierion $\phi$, the problem of $\phi$-optimal design generation has been addressed. A methodology to support the decision whether to continue or stop the search for optimal designs has been developed. It combines recent advances on discovery probability estimation, based on a Bayesian non parametric approach, \cite{favaro:12}, with well known methods for optimal design generation. 

In principle, this methodology could be applied to any discrete optimisation problem. This topic will be part of future research.

A software code, written in SAS, that makes use of the Proc Optex procedure, has been developed.

\begin{acknowledgements}
I would like to thank both Mauro Gasparini (Politecnico di Torino), \cite{gasp_metron},  and Giovanni Pistone (Collegio Carlo Alberto, Moncalieri, Torino) for the helpful discussions I had with them.
\end{acknowledgements}

% BibTeX users please use one of
%\bibliographystyle{spbasic}      % basic style, author-year citations
%\bibliographystyle{spmpsci}      % mathematics and physical sciences
%\bibliographystyle{spphys}       % APS-like style for physics
%\bibliography{fb}   % name your BibTeX data base

% Non-BibTeX users please use
\begin{comment}

\end{comment}

\end{document}